\documentclass[10pt,journal]{IEEEtran}
\usepackage{cite}
\usepackage{amsmath,amssymb,amsfonts,bbm,bm}
\usepackage{algorithm} 
\usepackage[noend]{algpseudocode} 
\usepackage{graphicx}
\usepackage{float}
\usepackage{subfigure}
\usepackage{textcomp}
\usepackage{xcolor}
\usepackage{float}
\usepackage{paralist}
\usepackage{multirow}
\usepackage{booktabs}
\usepackage[font=small]{caption}
\usepackage{graphicx}
\usepackage{tikz,pgfplots}
\usetikzlibrary{shapes.geometric, arrows, positioning}
\usepackage[nolist]{acronym}
\usepackage{epstopdf}
\usepackage{amsmath}
\usepackage{orcidlink}
\DeclareMathOperator*{\argmin}{argmin}

\newtheorem{definition}{Definition}

\newtheorem{theorem}{Theorem}

\newtheorem{preposition}{Preposition}

\newcommand{\mH}{\mathbf{H}}

\newtheorem{lemma}{Lemma}

\newcommand{\bv}{\mathbf{v}}
\newcommand{\bu}{\mathbf{u}}

\newcommand{\bs}{\mathbf{s}}
\newcommand{\bx}{\mathbf{x}}

\newcommand{\setC}{\mathbb{C}}
\newcommand{\setR}{\mathbb{R}}

\newcommand{\trp}{\mathsf{T}}
\newcommand{\her}{\mathsf{H}}
\newcommand{\set}[1]{\left\lbrace #1 \right\rbrace}

\newcommand{\brc}[1]{\left( #1 \right)}
\newcommand{\dbc}[1]{\left[ #1 \right]}

\newcommand{\abs}[1]{\left\vert #1 \right\vert}

\newcommand{\Ex}[1]{\mathbbm{E} \left\lbrace #1 \right\rbrace}

\def\BibTeX{{\rm B\kern-.05em{\sc i\kern-.025em b}\kern-.08em
    T\kern-.1667em\lower.7ex\hbox{E}\kern-.125emX}}
\begin{document}

\begin{acronym}
\acro{bs}[BS]{base station}
\acro{pa}[PA]{pinching antenna}
\acro{pass}[PASS]{pinching-antenna system}
\acro{mimo}[MIMO]{multiple-input multiple-output}
\acro{fl}[FL]{federated learning}
\acro{los}[LoS]{line-of-sight}
\acro{fdma}[FDMA]{Frequency Division Multiple Access}
\acro{tdma}[TDMA]{Time Division Multiple Access}
\acro{ps}[PS]{parameter server}
\acro{oma}[OMA]{orthogonal multiple access}
\acro{noma}[NOMA]{non-orthogonal multiple access}
\acro{irs}[IRS]{reconfigurable intelligent surface}
\acro{snr}[SNR]{signal to noise ratio}
\acro{sinr}[SINR]{signal to interference and noise ratio}
\acro{bcd}[BCD]{block coordinate descent}
\acro{rzf}[RZF]{regularized zero-forcing}
\acro{pas}[PAS]{pinching antenna system }
\acro{fp}[FP]{fractional programming}
\acro{mrt}[MRT]{maximal-ratio transmission}
\acro{zf}[ZF]{zero-forcing}
\acro{mse}[MSE]{mean square error}
\end{acronym}

\title{ Pinching Antennas-Assisted Low-Latency Federated Learning  Over Multi-User Wireless Networks}
\author{
    Saba Asaad$^{*}$\orcidlink{0000-0002-7959-8329},
    Hina~Tabassum\orcidlink{0000-0002-7379-6949},~\IEEEmembership{Senior~Member,~IEEE,} and  Ping Wang\orcidlink{0000-0002-7379-6949},~\IEEEmembership{Fellow,~IEEE,}

\thanks{
        S. Asaad, H. Tabassum, and P. Wang are with the Department of Electrical Engineering and Computer Science, York University, Toronto, Canada (e-mails:  
        \{\href{hinat@yorku.ca}{hinat}, \href{pingw@yorku.ca}{pingw}\}@yorku.ca). 
        
	}%

}

\maketitle
\begin{abstract}
	Federated learning (FL) over wireless networks is fundamentally constrained by unreliable communication links, particularly when uplink channels suffer from blockage, fading, or weak line-of-sight (LoS) conditions. Pinching-antenna systems (PASSs) offer a new physical-layer capability to dynamically reposition radiating points along a dielectric waveguide, enabling controllable LoS connectivity and significantly improved channel quality. This paper develops FedPASS, a novel framework for low-latency wireless FL assisted by PASS. 
	We formulate a multi-objective optimization problem that jointly minimizes the end-to-end round latency and an upper bound on the FL optimality gap. The resulting formulation is a mixed-integer nonlinear program subject to practical constraints on scheduling, transmit power, local CPU frequency, and PA placement. To address the resulting computational challenges, we develop a two-tier iterative algorithm: an outer loop that updates  scheduling, communication time allocation, and power control via block coordinate descent, and an inner loop that optimizes PA locations using a Gauss–Seidel–based coordinate update with grid search under spacing constraints. Numerical results on MNIST and CIFAR-10 demonstrate that FedPASS achieves accuracy comparable to idealized FL baselines while drastically reducing the total training latency compared to conventional wireless FL.
\end{abstract}
\begin{IEEEkeywords}
    Federated learning, pinching antennas, multi-objective optimization
\end{IEEEkeywords}
\raggedbottom
\section{Introduction}
The performance of wireless networks has been fundamentally constrained by the characteristics of the propagation channel, which were long regarded beyond human control. However, recent developments in wireless communications have challenged this paradigm by introducing mechanisms that enable dynamic reconfiguration of channel properties. In particular, flexible and reconfigurable antenna architectures have opened new possibilities for manipulating the propagation environment \cite{basharat2021reconfigurable}. Technologies such as reconfigurable intelligent surfaces (RISs) \cite{bereyhi2023channel}, fluid antennas \cite{wong2020fluid}, and movable antenna systems \cite{zhu2023movable} provide substantial control over the end-to-end channel by tuning their physical configurations. Specifically, RISs employ programmable phase shifts to reshape reflected signal paths \cite{bereyhi2022should}, while fluid-antenna and movable-antenna systems adjust the position of radiating elements within a confined region to realize favorable propagation \cite{renzo2019smart, wu2021intelligent}. 

Despite their advantages, both fluid and movable antenna designs restrict element displacement to only a few wavelengths, offering limited influence on large-scale path loss. As a result, when a direct line-of-sight (LoS) link is blocked, such limited repositioning range is often insufficient to restore the LoS component. This issue becomes more pronounced at higher carrier frequencies, where the shorter wavelengths further reduce the effective repositioning. In this context, the emerging pinching-antenna technology provides a solution with the objective of preserving strong LoS connectivity between transmitters and receivers \cite{ding2025flexible, suzuki2022pinching}. 

A \ac{pass} employs a long dielectric waveguide that can span tens of meters, enabling flexible placement of \acp{pa} without requiring dedicated RF chains or supplementary hardware.
Each \ac{pa} is formed by attaching a small dielectric element to the waveguide, thereby activating that location as a controlled radiation point. By selectively activating different points along the waveguide, \ac{pass} enables  reconfiguration of the propagation environment with high precision. This allows PASS to maintain reliable LoS connectivity even in scenarios where conventional antenna-based solutions fail. Furthermore, the ability to flexibly tune both the number and placement of active PAs, combined with a low-cost modular structure that enables effortless antenna addition, removal, or repositioning, makes PASS a highly versatile and cost-effective solution for next-generation wireless edge learning, distributed learning and computing solutions \cite{wang2025pinching, liu2025pinching}.

\par Wireless edge learning is inherently constrained by the  latency of uplink communications. In \ac{fl}, multiple devices collaboratively train a shared model without exchanging raw data \cite{konevcny2016federated, bereyhi2023device}. Each training round involves (i) uplink model aggregation, where devices transmit local updates to the server for weighted averaging, and (ii) downlink broadcasting of the updated global model \cite{hamidi2025rate}. By communicating model parameters rather than raw datasets, \ac{fl} preserves privacy and significantly reduces backhaul load. However, the frequent uplink and downlink exchange inherent in FL renders its performance highly sensitive to the stochastic nature of wireless channels. In particular, unreliable or weak uplink channels can delay aggregation, induce stragglers, and ultimately slow convergence and degrade model accuracy. PASS enables dynamic repositioning of radiation points closer to scheduled devices, thereby compensating for large-scale path-loss and stabilizing uplink transmission quality. This reduces latency of FL per round, and expands the set of devices that can reliably contribute to model aggregation, influencing both communication efficiency and statistical convergence.

In this work, we analyze the feasibility and significance of FedPASS, a  \ac{pass}-enabled FL framework, with the objective of jointly maximizing learning performance and minimizing communication latency, compared to conventional~approaches.




\subsection{Background Work}
\subsubsection{Related Work on PASS}
Early investigations established the fundamental signal behavior of PASS, examining how dielectric-waveguide propagation and localized PA activation work together to shape the radiation mechanism \cite{ding2024flexible}. 
Subsequent studies have demonstrated that PASS achieves superior performance compared to conventional fixed-antenna systems and current reconfigurable-antenna architectures.  For instance, by optimizing the number and spatial spacing of PAs along the dielectric waveguide, PASS has been shown to achieve significantly higher array gains than traditional antenna array \cite{ouyang2025array}. In addition, \ac{noma}-assisted PASS configurations with dynamically activated PAs have been shown to deliver markedly higher sum rates than their fixed-antenna counterparts \cite{wang2025antenna}. Beyond rate enhancement, physics-based modeling and joint transmit–pinching beamforming strategies developed in \cite{wang2025modeling} reveal that PASS can reduce the required transmit power significantly relative to conventional massive \ac{mimo} systems while maintaining comparable performance. In \cite{asaad2025dynamic}, the authors investigated the energy efficiency of PASS in a point-to-point communication setting and showed that dynamic tuning of pinching locations can significantly boost system efficiency. \color{black} PASS has also demonstrated notable advantages in secure communications, consistently outperforming fixed-location antenna architectures \cite{zhu2025pinching, wang2025pinching, badarneh2025physical}.
  
\par In addition to modeling and performance analysis, PA position optimization has become a central research focus in PASS. Several low-complexity or data-driven optimization strategies have been proposed in the literature. For example, probability-learning–based algorithms for PA placement were introduced in \cite{chen2025dynamic,asaad2025dynamic}, while a two-stage position adjustment approach was developed in \cite{xu2025rate} to iteratively refine PA locations. In addition, \cite{xie2025low} proposed a low-complexity placement design to maximize the downlink sum-rate for multiple users. More recently, a graph neural network (GNN)-enabled joint optimization framework  for PA placement and power allocation was explored in \cite{xie2025graph}, demonstrating the potential of learning-based methods in enhancing scalability.


\subsubsection{FL over \ac{pass}}
The integration of PASS into \ac{fl} has only recently been explored in a handful of research studies. In \cite{lin2025pinching}, the authors demonstrated that PASS can fundamentally alleviate the wireless straggler bottleneck in FL by dynamically shortening worst-case uplink distances. They analyzed synchronous and asynchronous FL and showed that PASS mitigates stragglers, reduces latency tails, and improves convergence. This work primarily focused on straggler mitigation. In \cite{wu2026straggler}, a hybrid conventional and pinching antenna network was developed for straggler-resilient FL, where fuzzy-logic client classification and DRL-based optimization were used to minimize latency under NOMA transmission. In \cite{overtheAir2026}, PASS was investigated for analog over-the-air FL, where PASS placement, scheduling, and power scaling were optimized under a computation-rate framework for AirComp-based aggregation. This work fundamentally targets analog aggregation efficiency and does not address the interaction between PA placement and the latency–learning tradeoff in digital FL. 
\color{black}
\subsection{Contributions}
Despite recent advances, a unified optimization framework that explicitly couples PA placement with the learning dynamics of digital \ac{fl} remains largely unexplored. Existing studies primarily address either physical-layer straggler mitigation or analog over-the-air aggregation, without jointly characterizing how antenna reconfiguration impacts both training latency and statistical convergence performance.  As a result, the interplay between communication design and learning performance in PASS-enabled digital FL is not yet fully understood.
To the best of our knowledge, although prior works have investigated PASS-assisted federated learning, none has developed a unified optimization framework that explicitly captures the latency–learning tradeoff through joint PA placement, device scheduling, and wireless resource control. Accordingly, the contributions of this article are summarized as follows:

$\bullet$  We propose a novel framework, namely FedPASS, that integrates \ac{pass} into wireless \ac{fl} to enhance uplink reliability and aggregation quality. By optimizing PA placement to strengthen uplink channels, FedPASS enables more effective device participation and reduces aggregation error, thereby accelerating FL convergence. However, increased participation also raises communication latency, necessitating a principled latency–learning tradeoff design.

$\bullet$ We formalize the inherent latency–learning tradeoff as a multi-objective optimization problem that jointly minimizes communication–computation latency in each FL round and a {learning convergence} metric referred to as \textit{optimality gap}. Particularly, this metric quantifies the amount of training data excluded in each FL round and captures how partial device participation affects convergence behavior of FL. 
\color{black} The resulting formulation is a mixed-integer nonlinear program (MINLP) involving discrete scheduling and PA-position variables coupled with continuous power-control decisions. 
Solving this problem characterizes the Pareto frontier of achievable latency–learning tradeoffs in PASS-enabled FL systems.

$\bullet$   We develop a  two-tier iterative algorithm tailored to the MINLP structure. The outer tier employs \ac{bcd} to optimize scheduling, communication time allocation, and power control through a sequence of convex subproblems. The inner tier addresses the non-convex PA-placement problem using a Gauss–Seidel coordinate-update method that enforces minimum antenna-spacing constraints. This hierarchical decomposition yields a scalable and computationally efficient approach for solving PASS-assisted FL optimization. We further provide computational complexity characterization of the proposed algorithm.

$\bullet$ Through extensive experiments on MNIST and CIFAR-10 datasets, we demonstrate that the proposed FedPASS framework achieves a final test accuracy closely matching the Perfect FL benchmark and consistently attains higher accuracy across communication rounds compared to the conventional fixed-antenna scheme. Moreover, FedPASS reduces training latency by up to {$6.4\times$} relative to conventional wireless FL schemes. 

\subsection{Paper Organization and Notations}
The rest of this paper is organized as follows: In Section II, we present the system model and assumptions. Section III  provides a formulation to the design problem and casts it as a multi-objective optimization. We develop a two-tier iterative algorithm to solve the design problem in Section IV. Section V provides numerical validation of the proposed scheme followed by the conclusion in Section VI. 

Notation: Scalars, vectors, and matrices are denoted by non-bold, bold lower-case, and bold upper-case letters, respectively. The transpose, conjugate and conjugate transpose of $\mH$ are denoted by $\mH^{\trp}$, $\mH^*$ and $\mH^{\her}$, respectively. The sets $\setR$ and $\setC$ are the real axis and the complex plane, respectively. $\mathcal{CN}\brc{\eta,\sigma^2}$ is the complex Gaussian distribution with mean $\eta$ and variance $\sigma^2$. $\dagger$ denotes the pseudo-inverse. $\Ex{.}{}$ denotes the expectation. $\{1,\ldots,N\}$ is abbreviated as $[N]$ and~$\nabla$ denotes the gradient operator.
\section{System Model and Assumptions}
\label{sec:system_model}
As illustrated in Fig. \eqref{Figure: Schematic}, we consider a wireless network comprising a \ac{ps} equipped with a \ac{pass} that serves $K$ single-antenna edge devices. This work considers an \ac{oma}-based transmission strategy. Specifically, we adopt \ac{tdma}, where the $K$ users are served in separate time slots. We assume that the devices are deployed in a rectangular area on the $x$-$y$ plane with dimensions $D_x \times D_y$. The position of the $k$-th device is denoted by $\psi_k = [x_k^{\text{U}}, y_k^{\text{U}}, 0]$, where $0 \leq x_k^{\text{U}} \leq D_x$ and $-D_y/2 \leq y_k^{\text{U}} \leq D_y/2$. The \ac{ps} is equipped with $N$ active \acp{pa} mounted on a waveguide aligned parallel to the $x$-axis at height $d$, with the $n$-th \ac{pa} at $\phi_n = [x_n^{\rm P}, 0, d]$, where $x_n^{\rm P} \in [0, D_x]$ denotes the location of the \ac{pa} on the waveguide. The PS feed-point is at $\phi_0 = [0, 0, d]$.  
\begin{figure}[t]
\centering
\includegraphics[height=0.3\textwidth]{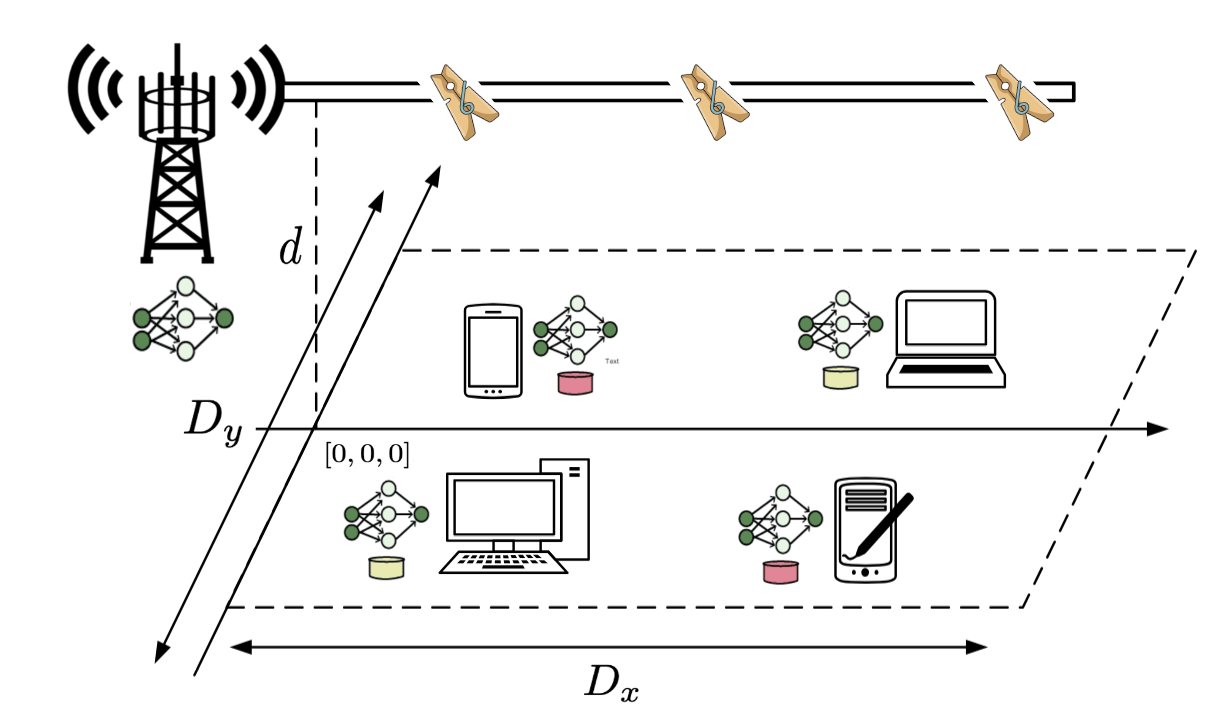}
\caption{Graphical illustration of the considered system model for FL over pinching antenna system.}
\label{Figure: Schematic}
\end{figure}
\subsection{Over-the-Air FL Setting}
The devices aim to train a global model on their distributed data using the \ac{fl} framework. More precisely, we consider a supervised learning setting, in which each device indexed by $k$ possesses a local training dataset of size $|\mathcal{D}_k|$, represented as $\mathcal{D}_k = \{ (\mathbf{u}_{k,i}, v_{k,i}) : 1 \leq i \leq |\mathcal{D}_k| \}$, where $\mathbf{u}_{k,i} \in \setR^D$ is the $i$-th data sample, and $v_{k,i}\in \setR$ is its corresponding label. The goal of the system is to collaboratively train a global model $G_{\boldsymbol{\omega}}: \setR^D \mapsto \setR$ over the entire dataset, i.e.,
$
    \mathcal{D} = \mathcal{D}_1 \cup \ldots \cup \mathcal{D}_K,
$
without sharing the local datasets over the network. 

For this model, the local empirical loss that can be computed on $\mathcal{D}_k$ at device $k$ is given by:
\begin{equation}
F_k(\boldsymbol{\omega} \vert \mathcal{D}_k) = \frac{1}{|\mathcal{D}_k|} \sum_{ (\mathbf{u}_k, v_k) \in \mathcal{D}_k } \ell (G_{\boldsymbol{\omega}} (\mathbf{u}_{k}), v_{k}),
\end{equation}
where $\ell(\cdot)$ denotes the loss function used for training, e.g., cross-entropy. The global loss function is then defined as:
\begin{equation}
F(\boldsymbol{\omega} \vert \mathcal{D}) = \sum_{k=1}^{K} \frac{|\mathcal{D}_k|}{|\mathcal{D}|}  F_k(\boldsymbol{\omega} \vert \mathcal{D}_k).
\end{equation}
The global training on $\mathcal{D}$ can then be formulated as 
\begin{align}\label{objective}
    \min_{\omega} F(\boldsymbol{\omega} \vert \mathcal{D}),
\end{align}
which clearly needs the knowledge of local datasets at the \ac{ps}.

\ac{fl} uses an iterative approach to solve the global training problem in \eqref{objective}. Each iteration of the algorithm is referred to as a \textit{communication round}. In the $t$-th communication round, the following steps are performed:
\begin{enumerate}
    \item \textit{Global Model Broadcasting:} At the beginning of the $t$-th communication round, the PS broadcasts the global model parameter vector $\boldsymbol{\omega}_{t-1}$ to all participating  devices.
    \item \textit{Local Model Computation:} After receiving the global model $\boldsymbol{\omega}_{t-1}$, each device performs $\vartheta$ steps of a gradient-based optimization algorithm, e.g., stochastic gradient descent (SGD), to update its local model using its local dataset. Specifically, device $k$ initializes its local parameter as $\boldsymbol{\omega}_k^{(0)} = \boldsymbol{\omega}_{t-1}$, and iteratively updates it according to $\boldsymbol{\omega}_k^{(j+1)} = \boldsymbol{\omega}_k^{(j)} - \zeta \nabla F_k\left( \boldsymbol{\omega}_k^{(j)} \mid \mathcal{D}_k \right)$ for $j = 0, 1, \ldots, \vartheta - 1$ where $\zeta$ is the learning rate. After completing the $\vartheta$ local update steps, the local model is set to $\boldsymbol{\omega}_{k,t} = \boldsymbol{\omega}_k^{(\vartheta)}$.
     \item \textit{Device Scheduling and Uplink Transmission:} The \ac{ps} performs scheduling and selects a subset of devices to participate in each communication round. Let $\mathcal{S} \subseteq [K]$ denote the subset of devices that are scheduled to share their model in round $t$. We can represent this set using a \textit{scheduling mask} variable: for each device $k \in \mathcal{S}$, the scheduling mask $s_k^{t} \in \{0, 1\} $ represents its activity in iteration $t$, where $ s_k^{t} = 1 $ indicates that the device is scheduled for the $t$-th round, and $s_k^{t} = 0$ otherwise. The selected devices upload their local model parameters $ \boldsymbol{\omega}_{k, t}$ to the \ac{ps} over the wireless channel. After receiving the local models from all scheduled devices, the \ac{ps} aggregates them to update the global model as:
\begin{align}
\boldsymbol{\omega}_{t} = \frac{1}{D(\bs^t)} \sum_{k=1}^{K} s_k^{t} |\mathcal{D}_k| \boldsymbol{\omega}_{k,t},
\end{align}
where $D(\bs^t) = \sum_{k=1}^{K} s_k^{t} | \mathcal{D}_k|$ for $\bs^t = [s_1^t, \ldots, s_K^t]$.
\end{enumerate}

\subsection{Communication Model}
The channel between the $k$-th device and the PS consists of two segments: a free-space channel from the device to the \acp{pa} and a guided channel along the waveguide to the \ac{ps}. The free-space channel for device $k$ can be represented by a linear model with channel vector $\mathbf{h}_{k} (\bx)\in\setC^N$ being given by:
\begin{equation} \label{eq:freespacechannel}
\mathbf{h}_{k} (\bx) = \sqrt{\eta}
\left[ \frac{ e^{-j \frac{2\pi}{\lambda} \|\psi_k - \phi_1\| }}{\|\psi_k - \phi_1\|}, \ldots, \frac{e^{-j \frac{2\pi}{\lambda} \|\psi_k - \phi_N \|}}{\|\psi_k - \phi_N \|} \right]^\trp ,
\end{equation}
which is a function of the pinching antenna positions  $\bx=[x_1^{\rm P}, \cdots, x_N^{\rm P}]$\color{black}, $ \eta = c^2 / (16 \pi^2 f_c^2)$ with $c$ and $f_c$ being the speed of light and carrier frequency, respectively, and $\lambda = c/f_c$ is the wavelength. The term $\|\psi_k - \phi_n \|$ denotes the Euclidean distance between device $k$ and \ac{pa} $n$. The guided channel from \acp{pa} to the \ac{ps} feed-point is further given by:
\begin{equation}\label{waveguidchannel}
\mathbf{g} (\bx) = \left[ e^{-j \frac{2\pi}{\lambda_g} \|\phi_0 - \phi_1\|}, \ldots, e^{-j \frac{2\pi}{\lambda_g} \|\phi_0 - \phi_N \|} \right]^\trp ,
\end{equation}
where $\lambda_g\!=\!\lambda / n_\mathrm{neff}$ is the g\textit{uided wavelength} with $n_\mathrm{neff}$ denoting the effective refractive index of the dielectric waveguide. 

Given \eqref{eq:freespacechannel} and \eqref{waveguidchannel}, 
the effective uplink channel from device $k$ to PS through \ac{pass} can be compactly represented as
\begin{equation}
H_k(\bx) = \mathbf{g}^\trp(\bx)\mathbf{h}_k(\bx).
\end{equation}
Following standard PA-assisted transmission models \cite{ding2025flexible}, we adopt a TDMA protocol for uplink, as illustrated in Fig.~\ref{Figure: Sum_Rate_User}. Under this protocol, each communication round is divided into orthogonal time slots allocated to the scheduled devices.

\section{Evaluation Metrics and Problem Formulation}
In this section, we present the end-to-end training latency and energy consumption models for the considered PASS-assisted FL system, as well as the metric for measuring the learning performance in the \ac{fl} setting. Invoking these models, we formulate the multi-objective optimization problem. 
\subsection{End-to-End Training Latency Model}
\begin{figure}[t]
\centering
\includegraphics[height=0.3\textwidth]{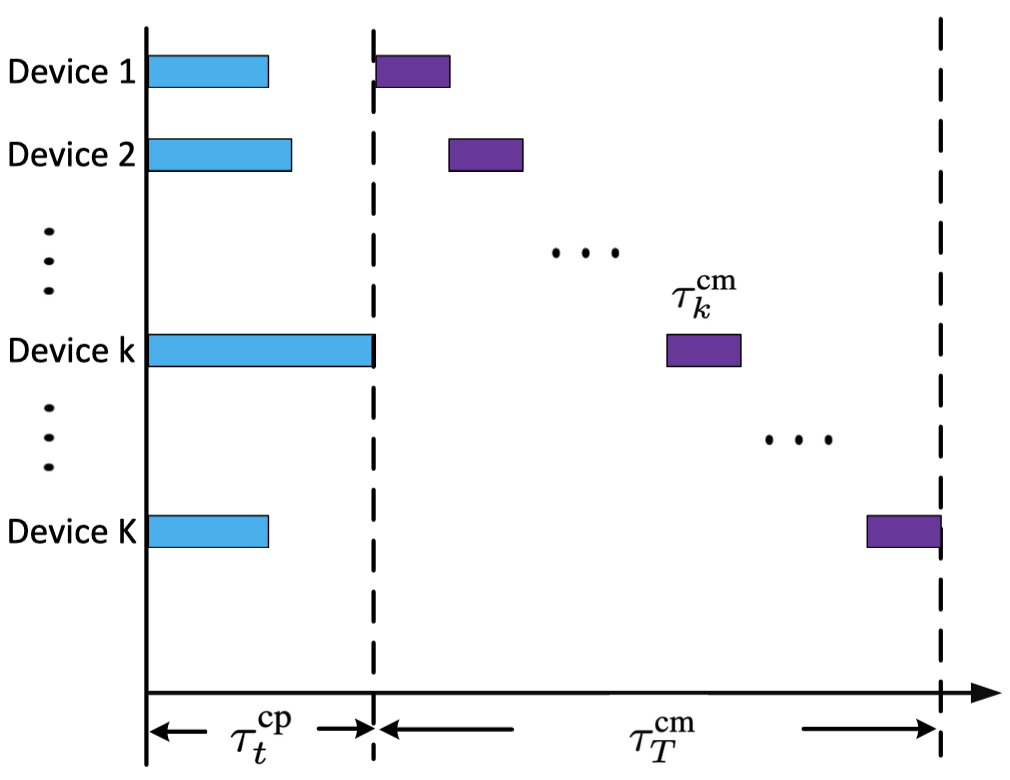}
\caption{Time line of FL-assisted PASS based on TDMA.}
\label{Figure: Sum_Rate_User}
\end{figure}
In a typical \ac{fl} setup, each communication round involves multiple steps carried out by both edge devices and the \ac{ps}. The total latency of a round, defined as the time required to complete one full cycle of training and communication, is hence influenced by two primary factors: the time required for local computation at each device, and the time required for model exchange over the wireless channel. Consequently, we decompose the total latency into two fundamental components: ($i$) \textit{computation latency}, which captures the time consumed for local training at the edge devices, ($ii$) \textit{communication latency}, which accounts for the time required to exchange model parameters between the \ac{ps} and the devices. 
Due to the high transmit power at the server, the downlink transmission delay is negligible compared to the uplink transmission time. Therefore, in our analysis, we focus primarily on the latency associated with local model computation and uplink transmission of model parameters \cite{yang2020energy, zhang2024joint}.\color{black}
\subsubsection{Local Computation Latency} Let \( C_k \) denote the number of CPU cycles required to process a single data sample on device \( k \), which can be estimated beforehand through offline profiling. The total computational workload for a single local training round on device \( k \) amounts to \( C_k \vert\mathcal{D}_k\vert \) CPU cycles. If the processing frequency of device \( k \) is given by \( f_k \), then, the corresponding computation time of device $k$ is given by:
\begin{equation}
    \tau_k^{\text{cp}} = \frac{C_k \vert\mathcal{D}_k\vert}{f_k}. 
\end{equation}
Assuming synchronous training across all participating devices, which is consistent with common practices in \ac{fl} (e.g., \cite{bouzinis2023wireless, amiri2020federated}), the overall computation latency in the \( t \)-th communication round is determined by the slowest device among those scheduled to participate. More precisely, the computation latency of communication round $t$ is given by:
\begin{equation}
    \tau_t^{\text{cp}} = \max_{k \in \mathcal{S}} \{s_k^t \tau_k^{\text{cp}}\},
\end{equation}
where  \( s_k^t \in \{0, 1\} \) is the scheduling mask of device \( k \) in the communication round \( t \).
\subsubsection{Communication latency}Since the local model has a fixed dimension across devices, the data size that each device must upload remains constant and is denoted by \( D_b \) bits. To ensure successful transmission of the local model during the communication phase, device $k$ requires to meet
\begin{align}
    \dfrac{D_b s_k^t}{R_k}\leq \tau_k^\text{cm}, 
\end{align}
where $\tau_k^\text{cm}$ denotes the communication time allocated to device $k$, and $R_k$ represents its achievable data rate, given by
\begin{equation}
R_k  = 
B \log_2 \left( 1 + \frac{P_k \eta }{N \sigma^2} G_k (\bx) \right),
\end{equation}
where $B$ denotes the signal bandwidth, $P_k$ is the transmit power of device $k$, and $\sigma^2$ is the noise variance at the PS. Moreover,
\begin{align}
    \hspace{-0.2cm}\!\!G_k (\bx)\!=\!\vert H_k(\bx) \vert^2\!=\!\left| \sum_{n=1}^N \frac{e^{-j \left( \frac{2\pi}{\lambda} \|\psi_k - \phi_n\| + \frac{2\pi}{\lambda_g} \|\phi_0 - \phi_n \| \right)}}{\|\psi_k - \phi_n \|} \right|^2\!\!\!,\label{gequation}
\end{align}
being the channel gain from device $k$ to the PS through the pinching antenna system.
As a result, the total communication time across all selected devices is given by
$
    \tau^\text{cm}_t = \sum_{k=1}^{K} s_k^t\tau_k^\text{cm},
$
which models the total communication latency in iteration $t$.
The end-to-end latency in the $t$-th communication round  can thus be given by:
\begin{align}\label{tau}
    \tau_t  &=  \tau^\text{cm}_t + \tau_t^{\text{cp}} =\sum_{k=1}^{K} s_k^t\tau_k^\text{cm} + \tau_t^{\text{cp}}.
\end{align}


\subsection{Energy Consumption Model}
Similar to latency, the total energy consumed at each device comprises two main components, i.e., (i) the energy required for local computation, and (ii) the energy spent for uploading the local model over the wireless network. 
\subsubsection{Energy for Local Computation} Following \cite{yang2020energy}, the energy consumed for local model training on device \( k \) can be modeled as follows:
\begin{equation}
    E_k^{\text{cp}} = \xi f_k^3 \tau_k^{\text{cp}} = \xi C_k | \mathcal{D}_k | f_k^2, 
\end{equation}
where \( \xi \) is a device-dependent constant reflecting hardware-specific characteristics. For sake of simplicity, in the sequel, we consider the same \( \xi \) across the local devices. 

\subsubsection{Energy for Local Model Uploading}
The energy consumption of device \( k \) for uploading its local model to the \ac{ps} in each round is given by
$
    E^{\rm cm}_{k}= P_k \tau_k^\text{cm},
$
with $\tau_k^\text{cm}$ being the communication time allocated to~device~$k$.

The total energy consumed in each communication round by the local device $k$ can thus be given by:
\begin{align}
    E_k &= E^{\rm cm}_{k} + 
     E_k^{\text{cp}}
     = P_k \tau_k^\text{cm} + \xi C_k |\mathcal{D}_k| f_k^2. 
\end{align}
\subsection{Learning Performance: Optimality Gap}
To evaluate the learning performance of the distributed FL setting, we use the notion of \textit{optimality gap}, which is standard in convex optimization \cite{zheng2023semi, cao2021optimized}. In our setting, the optimality gap is defined as
$
O_t = F(\boldsymbol{\omega}_t) - F(\boldsymbol{\omega}^\star),
$
where $F(\boldsymbol{\omega}_t)$ denotes the global loss after $t$ communication rounds,\footnote{For the sake of brevity, we drop $\mathcal{D}$ in the sequel.} and $\boldsymbol{\omega}^\star \in \arg\min_{\omega} F(\boldsymbol{\omega})$ represents a global minimizer of the empirical loss.
\color{black}
The characterization of the optimality gap depends on the underlying loss function and model architecture. In the considered PASS-assisted FL system, communication quality directly affects device scheduling and, consequently, convergence behavior. To obtain a tractable performance metric that captures this interaction, we derive an upper bound on the optimality gap under standard smoothness and gradient assumptions, following the analytical framework in \cite{xing2021federated, guo2022joint}.
\subsubsection{Bounding Optimality Gap}
To establish the upper bound, we impose the following commonly adopted assumptions.
\begin{itemize}
    \item[A.1] The local loss functions \( F_k(\boldsymbol{\omega} \vert \mathcal{D}_k) \) are continuously differentiable, i.e., \( \nabla F_k(\boldsymbol{\omega}) \) exists and is continuous. 
    \item[A.2] The gradient of the global loss \( F(\boldsymbol{\omega} \vert \mathcal{D}) \) is Lipschitz continuous with constant \( L \), i.e., 
\begin{align}
   \|\nabla F(\boldsymbol{\omega}_1) - \nabla F(\boldsymbol{\omega}_2)\|_2 \leq L \|\boldsymbol{\omega}_1 - \boldsymbol{\omega}_2\|_2, 
\end{align}
for any $\boldsymbol{\omega}_1 \neq \boldsymbol{\omega}_2$.
\item[A.3] The loss satisfies the Polyak-Łojasiewicz condition with parameter \( \delta \), which indicates 
\begin{align}
\|\nabla F(\boldsymbol{\omega})\|_2^2 \geq 2 \delta (F(\boldsymbol{\omega}) - F(\boldsymbol{\omega}^\star)),    
\end{align}
at any point $\boldsymbol{\omega}$. 
\item[A.4] The sample local gradients are bounded in norm with some positive constant \( \varepsilon \), i.e., 
\begin{align}
    \Vert\nabla \ell (G_{\boldsymbol{\omega}} (\mathbf{u}_{k}), v_{k}) \Vert_2 \leq \varepsilon,
\end{align}
for any $(\mathbf{u}_{k}, v_{k}) \in \mathcal{D}_k$.
\end{itemize}

Under the above assumptions, the decrease in the global empirical loss after each communication round can be bounded, as stated in the following lemma.
\begin{lemma}\label{lem:1}
    Let assumptions (A.1)-(A.4) hold. For the learning rate \( \zeta = 1/L \), the global loss in the communication round $t$  satisfies    
\begin{align} \label{lemma1}
F(\boldsymbol{\omega}_{t}) \leq &F(\boldsymbol{\omega}_{t-1}) - \frac{\vartheta}{2L} \|\nabla F(\boldsymbol{\omega}_{t-1})\|_2^2 + \frac{\vartheta}{2L} \|\mathbf{e}_{t-1}\|_2^2 \nonumber \\
&+ \frac{\vartheta^3 \varepsilon^2}{2 L|\mathcal{D}| D(\bs^t)},
\end{align}
where $\mathbf{e}_{t-1}$ is defined as
\begin{align}\label{eq:e}
\mathbf{e}_{t-1} = \nabla F(\boldsymbol{\omega}_{t-1}) - \frac{1}{D(\bs^t)} {\sum_{k=1}^K s_k^t |\mathcal{D}_k| \nabla F_k(\boldsymbol{\omega}_{t-1})},
\end{align}
and defines the aggregation error induced by partial device scheduling.
\end{lemma}
\begin{IEEEproof}
    Proof is given in \textbf{Appendix A}.
\end{IEEEproof}
Using Lemma~\ref{lem:1}, we can bound the optimality gap as given in the following theorem.
\begin{theorem}\label{thm:1}
    Under the assumptions (A.1)-(A.4), the optimality gap after \( T \) communication rounds is bounded as:
\begin{align}
    O_T \leq \left(1 - \frac{\delta \vartheta}{L}\right)^T O_0 + \sum_{t=1}^{T} A_t \left(1 - \frac{\delta \vartheta}{L}\right)^{T-t},
\end{align}
where $A_t$ is given by 
\begin{align} \label{At}
    A_t = \frac{2 \vartheta\varepsilon^2}{L |\mathcal{D}|^2} 
    \brc{
    \abs{\mathcal{D}} - D(\bs^t)
    }^2+ \frac{\vartheta^3 \varepsilon^2}{2 L|\mathcal{D}|^2},
\end{align}
and $O_0$ denotes the optimality gap in the initial round.
\end{theorem}

\begin{IEEEproof}
    Proof is given in \textbf{Appendix B}. 
\end{IEEEproof}
The first term in $A_t$ quantifies the aggregation error due to partial participation, while the second term captures the local-drift effect accumulated over $\vartheta$ local steps. See Appendix B for the detailed derivations and constants.

\subsubsection{Learning Performance}
The ultimate goal of distributed learning is to converge to the optimal global empirical risk. 
From Theorem~\ref{thm:1}, the optimality gap is upper bounded by a weighted sum of $A_t$. 
For a fixed learning setup, the quantities $L$, $\delta$, $\vartheta$, $\varepsilon$, and $|\mathcal D|$ are determined by the model, loss function, and local training configuration, and are therefore treated as constants. In this setting, the only term in $A_t$ that depends on the scheduling variable $\mathbf s^t$ is $D(\mathbf s^t)$. Hence, minimizing $A_t$ with respect to $\mathbf s^t$ is equivalent to minimizing 
\begin{align}
\argmin_{\bs^t} A_t = \argmin_{\bs^t} \abs{\mathcal{D}} - D(\bs^t).
\end{align}
This motivates defining the learning performance metric
\begin{align}\label{F_learn}
F_{\rm learn} = |\mathcal D| - D(\mathbf s^t),
\end{align}
which captures the scheduling-dependent component of the upper bound in Theorem~\ref{thm:1}. 
Minimizing $F_{\rm learn}$ therefore minimizes the scheduling-dependent component of the upper bound in Theorem~\ref{thm:1}. \color{black}


\subsection{Multi-Objective Problem Formulation}
The bound in Theorem~\ref{thm:1} characterizes the convergence behavior of the optimality gap across communication rounds. 
Let the total training time be defined as $\tau_T = \sum_{t=1}^{T} \tau_t$.
Minimizing $\tau_T$ while achieving a desired accuracy level requires jointly controlling both the per-round latency $\tau_t$ and the convergence behavior governed by $A_t$. Scheduling more devices per round improves the convergence rate by reducing $A_t$, but increases communication latency. 
Conversely, reducing per-round latency typically decreases the number of participating devices, which increases the aggregation error term $A_t$ and may slow the decay of the optimality gap, thereby requiring more rounds to reach a target accuracy. To capture this fundamental trade-off in a tractable manner, we formulate a per-round bi-objective optimization problem that jointly minimizes latency and the scheduling-dependent component of the optimality-gap bound:\color{black}
\begin{subequations}
\begin{align}
	{\mathcal{P}_1} &\min_{\bx, f_k, P_k, \bs^t, \boldsymbol{\tau}}\set{ \tau_t , F_{\rm learn}}\nonumber\\
	&\text{s.~t. } \; \;   \mathcal{C}_1: \dfrac{D_b s_k^t}{R_k} \leq \tau_k^{\text{cm}}, \quad \forall k \in [K], \\
 &\phantom{\text{s. t. }} \; \;  \mathcal{C}_2: P_k \tau_k^{\text{cm}} + \xi C_k |\mathcal{D}_k| f_k^2 \leq E_k^{\text{max}}, \forall k \in [K],\\
 &\phantom{\text{s. t. }} \; \;  \mathcal{C}_3: \bs^t \in \{0, 1\}^K, \quad \boldsymbol{\tau} \in \setR_+^K, \quad \forall k \in [K], \\
 &\phantom{\text{s. t. }} \; \;  \mathcal{C}_4: s_k^t \dfrac{C_k |\mathcal{D}_k| }{f_k} \leq \tau_t^\text{cp}, \quad \forall k \in [K],\\
  &\phantom{\text{s. t. }} \; \;  \mathcal{C}_5:  x_{n+1}^{\rm P} - x_{n}^{\rm P} \geq \Delta \ell  \quad \forall n \in [N], \\
    &\phantom{\text{s. t. }} \; \; \mathcal{C}_6: 0 \leq P_k \leq P, \; 0 \leq f_k \leq f_{\max}, \forall k \in [K], 
\end{align}
\end{subequations}
where the objectives $\tau_t$ and $F_{\rm learn}$ are defined in \eqref{tau} and \eqref{F_learn}, respectively. Furthermore, $\boldsymbol{\tau}=[\tau_1^\text{cm}, \cdots, \tau_K^\text{cm}, \tau_t^\text{cp}]^\trp$ and $\bx =[x_1, \ldots, x_N]^\trp$ denote the location of the pinching antennas on the waveguide. The constraints are given as follows:
\begin{itemize}
    \item[$\mathcal{C}_1$:] This constraint guarantees successful uplink model transmission by device $k$.  
    \item[$\mathcal{C}_2$:] It constrains the consumed energy at device $k$ to remain below the maximum energy budget. 
    \item[$\mathcal{C}_3$:] Constrains the support of scheduling mask $\bs^t$ and the scheduled time intervals, i.e., the computation and communication delays, to be in the feasible set.
    \item[$\mathcal{C}_4$:] This constraint guarantees that the allocated computation time $\tau_t^{\text{cp}}$ is larger than the maximum computation time of scheduled devices.
    \item[$\mathcal{C}_5$:] It constrains the minimum distance between two neighboring pinching antennas on the waveguide.
    \item[$\mathcal{C}_6$:] This constraint guarantees that the transmit power and computation clock frequency remain below the available power and frequency budget.
\end{itemize}

The optimization problem $\mathcal{P}_1$ is a MINLP due to the discrete nature of the scheduling mask \( \bs^t \), nonlinear constraints with tightly coupled optimization variables, and nonconvex objective functions.  

\section{Joint Latency-Learning Optimization in PASS-Assisted FL Systems}

\subsection{Equivalent Problem Reformulation}
Due to the competing  objectives in $\mathcal{P}_1$, i.e., minimizing communication latency and maximizing learning performance, the problem does not generally yield a unique global optimum. The classical approach is to define the so-called \emph{Pareto front}, which consists of Pareto-optimal solutions, where no objective can be improved without degrading the other.
\begin{definition} 
Let $\bu = \set{ \bx, f_1, \ldots,f_K, P_1, \ldots, P_K, \bs^t, \boldsymbol{\tau} }$ denote the design variables. A feasible solution \( \bu^\star \) is called \emph{Pareto optimal} if there does not exist another feasible solution \( \bu \) such that $\tau_t \leq \tau_t^\star$, $F_{\mathrm{learn}} \leq F_{\mathrm{learn}}^\star$ and at least one of the inequalities is strict. Here, $\tau_t^\star$ and $F_{\mathrm{learn}}^\star$ denote the objective values computed at $\bu^\star$. The \emph{Pareto front} is then defined as the set of all Pareto optimal objective pairs, i.e.,
$
\mathbb{P} = \left\{ \left( \tau_t, F_{\mathrm{learn}} \right): \bu \text{ is Pareto optimal} \right\}.
$
\end{definition}

Similar to single-objective optimization problems (SOOPs), Pareto-optimal solutions can be classified as local or global. In the convex setting, where both objective functions and the feasible region are convex, any local Pareto-optimal solution is also globally optimal. In the non-convex case, however, only local Pareto-optimal solutions can generally be computed using algorithmic approaches \cite{miettinen1999nonlinear}.  We adopt the scalarization approach, which allows us to convert the bi-objective problem into a single-objective formulation. This transformation enables algorithmic tractability while facilitating exploration of the Pareto front. Before applying scalarization, we first simplify the constraint set in $\mathcal{P}_1$ to reduce problem complexity. In particular, we begin by reformulating constraint $\mathcal{C}_4$ as a strict equality, followed by additional simplifications. Once the problem structure is streamlined, we proceed with the scalarization to obtain a tractable single-objective formulation.
\begin{preposition}
For any point in the Pareto front of problem $\mathcal{P}_1$, constraint $\mathcal{C}_4$ holds with equality.
\end{preposition}
\begin{IEEEproof}
Assume that an optimal solution $\{f_k^\star, \tau_t^{\text{cp}\star}\}$ satisfies constraint $\mathcal{C}_4$ with strict inequality for some device $k$, i.e., $\frac{s_k^t C_k |\mathcal{D}_k|}{f_k^\star} < \tau_t^{\text{cp}\star}$.
Now consider reducing the processing frequency to $\tilde{f}_k = \frac{s_k^t C_k |\mathcal{D}_k|}{\tau_t^{\text{cp}\star}}$.
Recalling that the computation energy at user $k$ is given by $E_k^{\text{cp}} = \xi C_k |\mathcal{D}_k| f_k^2$, the energy consumption at processing frequency $\tilde{f}_k$ becomes 
\begin{align}
\tilde{E}_k^{\text{cp}} = \xi C_k |\mathcal{D}_k| \tilde{f}_k^2 < \xi C_k |\mathcal{D}_k| (f_k^\star)^2 = E_k^{\text{cp}\star}.
\end{align}
 Since $\frac{s_k^t C_k |\mathcal D_k|}{f_k^\star} < \tau_t^{\rm cp\star}$, device $k$ is not the bottleneck in the computation phase. Hence, reducing $f_k$ decreases the computation energy while keeping $\tau_t^{\text{cp}}$ unchanged, so the overall latency $\tau_t$ and the learning metric $F_{\rm learn}$ remain unaffected and all constraints remain satisfied. Therefore, at any Pareto-optimal solution, constraint $\mathcal{C}_4$ cannot hold with strict inequality, and must hold with equality. \color{black}
\end{IEEEproof}

Based on Proposition 1, we can enforce $\mathcal{C}_4$ in $\mathcal{P}_1$ with equality. This yields $f_k = \frac{s_k^t C_k |\mathcal{D}_k|}{\tau_t^{\text{cp}}}$, for all $k \in [K]$. We next use the fact that $P_k = {E_k^{\text{cm}}}/{\tau_k^{\text{cm}}}$ to substitute the variable $P_k$ in $\mathcal{P}_1$ by $E_k^{\text{cm}} = P_k \tau_k^{\text{cm}}$. Under this variable exchange, the constraint $\mathcal{C}_2$ can be written as 
$
    E_k^{\text{cm}} + \xi s_k^t \frac{(C_k |\mathcal{D}_k|)^3}{(\tau_t^{\text{cp}})^2} \leq E_k^{\text{max}},
$
which leads to the following variational form:
\begin{subequations}
\begin{align}
	{\mathcal{P}_2} &\min_{\bx, E_k^{\rm cm}, \bs^t, \boldsymbol{\tau}}\set{ \tau_t , F_{\rm learn}},\nonumber\\
	&\text{s.~t. } \; \; \;  \bar{\mathcal{C}}_1: \dfrac{D_b s_k^t}{R_k} \leq \tau_k^{\text{cm}}, \\
 &\phantom{\text{s. t. }} \; \; \; \bar{\mathcal{C}}_2: E_k^{\text{cm}} + \xi s_k^t \frac{(C_k |\mathcal{D}_k|)^3}{(\tau_t^{\text{cp}})^2} \leq E_k^{\text{max}},\\
 &\phantom{\text{s. t. }} \; \; \; \bar{\mathcal{C}}_3: \bs^t \in \{0, 1\}^K, \quad \boldsymbol{\tau} \in \setR_+^K, \\
  &\phantom{\text{s. t. }} \; \; \; \bar{\mathcal{C}}_4:  x_{n+1}^{\rm P} - x_{n}^{\rm P} \geq \Delta \ell, \\
    &\phantom{\text{s. t. }} \; \; \; \bar{\mathcal{C}}_5: 0 \leq E_k^{\rm cm} \leq P\tau_k^{\rm cm}, \\
     &\phantom{\text{s. t. }} \; \; \; \bar{\mathcal{C}}_6: \frac{s_k^t C_k |\mathcal{D}_k|}{\tau_t^{\text{cp}}} \leq f_{\max}. 
\end{align}
\end{subequations}

\subsection{MOOP to SOOP Transformation}
Starting with $\mathcal{P}_2$, we reformulate the underlying MOOP as a single-objective problem, denoted as $\mathcal{P}_3$, in which a weighted average of the two objectives is to be optimized. More precisely, we define 
$
    F_\lambda \brc{\bu} = \lambda \tau_t  + \brc{1-\lambda} F_{\rm learn}
$
for some scalarization factor $\lambda \in (0,1)$. Using scalarization, weighted combinations of the two objectives allow us to explore the tradeoff between latency and learning performance within the feasible region. Hence, by sweeping $\lambda \in (0,1)$, we obtain a set of tradeoff solutions that provide an inner approximation of the Pareto front:
\begin{align}
	\mathcal{P}_3: &\min_{\bx,  E_k^{\rm cm}, \bs^t, \boldsymbol{\tau}} F_\lambda \brc{\bu} \nonumber
	&&\text{s.~t. } \; \;  (\bar{\mathcal{C}}_1)-(\bar{\mathcal{C}}_6). \nonumber
\end{align}
In the sequel, we solve $\mathcal{P}_3$ for a given $\lambda \in (0,1)$. Expanding $F_\lambda(\bu)$ yields
\begin{align}
    &F_\lambda \brc{\bu} = \lambda 
    (\sum_{k=1}^{K} s_k^t\tau_k^\text{cm} + \tau_t^{\text{cp}})
    + \brc{1-\lambda} (\abs{\mathcal{D}} - \sum_{k=1}^{K} s_k^{t} | \mathcal{D}_k|),\nonumber\\ \nonumber
    &=
    \lambda \tau_t^{\text{cp}} +
    \sum_{k=1}^{K} s_k^t \brc{\lambda \tau_k^\text{cm} - \brc{1-\lambda} | \mathcal{D}_k|}
    + \brc{1-\lambda} \abs{\mathcal{D}},\\ \nonumber
    &=
    \lambda \dbc{\tau_t^{\text{cp}} + \sum_{k=1}^{K} s_k^t \brc{\tau_k^\text{cm} + |\mathcal{D}_k|} } - D\brc{\bs^t}
    + \brc{1-\lambda} \abs{\mathcal{D}}.
\end{align}
We can hence re-write $\mathcal{P}_3$ as follows: 
\begin{align}
	\tilde{\mathcal{P}}_3: &\min_{\bu} \tilde{F}_\lambda \brc{\bu} 
	&&\text{s.~t. } (\bar{\mathcal{C}}_1)-(\bar{\mathcal{C}}_6),
\end{align}
where $\tilde{F}_\lambda$ is defined as
\begin{align}
    &\tilde{F}_\lambda \brc{\bu} =
    \lambda \dbc{\tau_t^{\text{cp}} + \sum_{k=1}^{K} s_k^t \brc{\tau_k^\text{cm} + |\mathcal{D}_k|} } - D\brc{\bs^t}.
\end{align}
To develop an algorithmic solution to the problem $\tilde{\mathcal{P}}_3$, we note that there are three major challenges which lead to computational hardness: 1) The scheduling mask is integer-valued. 2) The objective and constraints contain several cross-product terms which make them non-convex. 3) The term $G_k\brc{\bx}$ depends on the location of antennas through an extremely non-linear term. To address these challenges, we develop a two-tier iterative algorithm. Particularly, we address the integer problem by relaxing $\bs^t$ to the convex set $[0,1]^K$ and approximating the scheduling mask by thresholding the relaxed solution. This reduces the problem to a non-convex problem whose all marginals, i.e., optimization over a subset of variables in $\bu$ while the others are kept fixed, except the marginal on $\bx$, are convex. We hence develop an outer loop which uses block-coordinate descent (BCD) method to approximate the optimal design for a fixed $\bx$. Considering the highly-oscillating nature of $G_k\brc{\bx}$, we further develop an inner loop which handles the marginal optimization over $\bx$. The details are given in the next section.
\subsection{Proposed Two-Tier Block Coordinate Descent Algorithm}
We start our derivation by extending the scheduling mask $\bs^t$ to belong to the hypercube $[0,1]^K$. This relaxation transforms the original problem into a more tractable form. Based on this relaxed formulation, we design a two-tier iterative algorithm to progressively approximate the solution.
To this end, we divide the optimization variables as $\bu = \dbc{\bx, \bv}$, where $\bv = \set{  E^{\rm cm}_1, \ldots, E_K^{\rm cm}, \bs^t, \boldsymbol{\tau} }$. Using the BCD method, we approximate the solution of $\tilde{\mathcal{P}}_3$ by iterating between the solutions of the marginal problems on $\bx$ and $\bv$. Namely, let $\bx_i$ be the solution for $\bx$ in iteration $i$. We update $\bv$ to $\bv_{i+1}$ which is the solution to the following marginal problem:
\begin{subequations}
\begin{align}
	\mathcal{M}_1: &\min_{\bv} \tilde{F}_\lambda \brc{\dbc{ \bx_i, \bv }} \nonumber\\
	&\text{s.~t. } \; \; \;  \check{\mathcal{C}}_1: \dfrac{D_b s_k^t}{R_k} \leq \tau_k^{\text{cm}}, \\
 &\phantom{\text{s. t. }} \; \; \; \check{\mathcal{C}}_2: E_k^{\text{cm}} + \xi s_k^t \frac{(C_k |\mathcal{D}_k|)^3}{(\tau_t^{\text{cp}})^2} \leq E_k^{\text{max}},\\
 &\phantom{\text{s. t. }} \; \; \; \check{\mathcal{C}}_3: \bs^t \in [0,1]^K, \quad \boldsymbol{\tau} \in \setR_+^K, \\
    &\phantom{\text{s. t. }} \; \; \; \check{\mathcal{C}}_4: 0 \leq E_k^{\rm cm} \leq P\tau_k^{\rm cm},\\
    &\phantom{\text{s. t. }} \; \; \; \check{\mathcal{C}}_5: \frac{s_k^t C_k |\mathcal{D}_k|}{\tau_t^{\text{cp}}} \leq f_{\max}. 
\end{align}
\end{subequations}
where $\check{\mathcal{C}}_i$ denote the constraints obtained from $\bar{\mathcal{C}}_i$ after fixing $\mathbf{x}=\mathbf{x}_i$.
Using this solution, we update $\bx$ to $\bx_{i+1}$ by solving
\begin{subequations}
\begin{align}
	\mathcal{M}_2: &\min_{\bx} \tilde{F}_\lambda \brc{\dbc{ \bx, \bv_{i+1} }} \nonumber\\
	&\text{s.~t. } \; \; \;  \hat{\mathcal{C}}_1: \dfrac{D_b s_k^t}{R_k} \leq \tau_k^{\text{cm}}, \\
    &\phantom{\text{s. t. }} \; \; \; \hat{\mathcal{C}}_2:  x_{n+1}^{\rm P} - x_{n}^{\rm P} \geq \Delta \ell.
\end{align}
\end{subequations}
where $\hat{\mathcal{C}}_i$ denote the subset of constraints relevant to $\mathbf{x}$ obtained from $\bar{\mathcal{C}}_i$ after fixing $\bv=\bv_{i+1}$.
\subsection{Outer Loop}
To construct the inner loop, we first observe that the objective function is marginally linear in terms of $\bs^t$ and $\boldsymbol{\tau}$. Furthermore, we note that 
\begin{equation}
R_k  =  B \log_2 \left( 1 + \frac{E_k^{\rm cm} \eta }{ \tau_k^{\rm cm } \sigma^2} G_k (\bx) \right),
\end{equation}
is concave in $E_k^{\rm cm}$, while
$
\tau_k^{\rm cm } B
\log_2 \left( 1 + \frac{E_k^{\rm cm} \eta }{ \tau_k^{\rm cm } \sigma^2} G_k (\bx) \right),
$
is concave in $\tau_k^{\rm cm }$. This indicates that constraint $\check{\mathcal{C}}_1$ is marginally convex for $E_k^{\rm cm}$ and $\boldsymbol{\tau}$. Since $\tau_t^{\rm cp} > 0$ and $E_k^{\rm cm} < E_k^{\max}$ for feasible points, we can equivalently rewrite $\check{\mathcal{C}}_2$ as the following convex constraint (affine in $\tau_t^{\rm cp}$):
\begin{align}
\tau_t^{\rm cp}
\;\ge\;
\sqrt{
\xi s_k^t \frac{(C_k |\mathcal{D}_k|)^3}{E_k^{\text{max}} - E_k^{\text{cm}}}
}.
\end{align}
This suggests that we can address the marginal problem $\mathcal{M}_1$ via BCD approach, where we split the optimization variables as $\bv = \set{\boldsymbol{\tau}} \cup \set{\bs^t} \cup \set{E_1^{\rm cm}, \ldots, E_K^{\rm cm}}$. The inner problem for each of these variables describes a convex program. 
The first inner problem \color{black} is given by:
\begin{align}
	\mathcal{I}_1&: \min_{\boldsymbol{\tau}} 
    {\tau_t^{\text{cp}} + \sum_{k=1}^{K} s_k^t {\tau_k^\text{cm} } }
    \nonumber\\
	&\text{s.~t. }
    {D_b s_k^t} - \tau_k^{\text{cm}} B \log_2 \left( 1 + \frac{E_k^{\rm cm} \eta }{ \tau_k^{\rm cm } \sigma^2} G_k (\bx) \right) \leq 0, \\
 &\phantom{\text{s. t. }} \;
 \sqrt{
\xi s_k^t \frac{(C_k |\mathcal{D}_k|)^3}{E_k^{\text{max}} -E_k^{\text{cm}}}} - \tau_t^{\text{cp}} \leq 0,\\
&\phantom{\text{s. t. }} \; {s_k^t C_k |\mathcal{D}_k|} - f_{\max} {\tau_t^{\text{cp}}} \leq 0,
 \phantom{\text{s. t. }} \; \boldsymbol{\tau} \in \setR_+^K.
\end{align}
The problem is  convex in $\boldsymbol{\tau}$, as the objective is linear and all constraints are convex. 
We solve $\mathcal{I}_1$ via an interior-point method in CVX.
The second inner problem \color{black}is given by:
\begin{align}
	\mathcal{I}_2: &\min_{\bs^t } 
    {\tau_t^{\text{cp}} + \sum_{k=1}^{K} s_k^t {\tau_k^\text{cm} } }
    \nonumber\\
	&\text{s.~t. }
    {D_b s_k^t} - \tau_k^{\text{cm}} B \log_2 \left( 1 + \frac{E_k^{\rm cm} \eta }{ \tau_k^{\rm cm } \sigma^2} G_k (\bx) \right) \leq 0, \\
 &\phantom{\text{s. t. }} \;
 E_k^{\text{cm}} - E_k^{\text{max}} + \xi s_k^t \frac{(C_k |\mathcal{D}_k|)^3}{(\tau_t^{\text{cp}})^2} \leq 0,\\
 &\phantom{\text{s. t. }} \; {s_k^t C_k |\mathcal{D}_k|} - f_{\max} {\tau_t^{\text{cp}}} \leq 0,\\
 &\phantom{\text{s. t. }} \; \bs^t \in [0,1]^K.
\end{align}
With $(\boldsymbol{\tau}, \{E_k^{\rm cm}\})$ fixed, the problem is linear in $\bs^t$, and can be solved via linear programming.

(iii) The third inner \color{black}problem depends only through the constraints on the optimization variables. We hence find a solution by a point in the Pareto front of the following MOOP
\begin{align}
	\mathcal{I}_3: &\max_{E_k^{\rm cm} } \set{R_1, \ldots, R_K}
    \nonumber\\
	&\text{s.~t. }
 E_k^{\text{cm}} + \xi s_k^t \frac{(C_k |\mathcal{D}_k|)^3}{(\tau_t^{\text{cp}})^2} - E_k^{\text{max}}\leq 0,\\
 &\phantom{\text{s. t. }} \; 0 \leq E_k^{\rm cm} \leq P\tau_k^{\rm cm}. 
\end{align}
Noting that this problem is convex, one can readily find the Pareto front by solving the following sum-rate optimization:
\begin{align}
	\tilde{\mathcal{I}}_3: &\max_{E_k^{\rm cm} } \sum_{k=1}^K \theta_k R_k
    \nonumber\\
	&\text{s.~t. }
 E_k^{\text{cm}} + \xi s_k^t \frac{(C_k |\mathcal{D}_k|)^3}{(\tau_t^{\text{cp}})^2} - E_k^{\text{max}}\leq 0,\\
 &\phantom{\text{s. t. }} \; 0 \leq E_k^{\rm cm} \leq P\tau_k^{\rm cm}.
\end{align}
The scalarized problem is convex because $R_k$ is concave in $E_k^{\rm cm}$ and all constraints are affine.  
We solve $\tilde{\mathcal{I}}_3$ using CVX.
 \subsection{Inner Loop}\color{black}
To construct the inner loop, we note that  the objective in the marginal problem $\mathcal{M}_2$ depends on $\bx$ only through  $\hat{\mathcal{C}}_1$, so a solution can be found by solving the following:
\begin{subequations}
\begin{align}
	\tilde{\mathcal{M}}_2: &\max_{\bx} \set{R_1, \ldots, R_K} \nonumber\\
	&\text{s.~t. } \; \; \; \hat{\mathcal{C}}_2:  x_{n+1}^{\rm P} - x_{n}^{\rm P} \geq \Delta \ell,
\end{align}
\end{subequations}
which can be scalarized as follows:
\begin{subequations}
\begin{align}
	\hat{\mathcal{M}}_2: &\max_{\bx} \sum_{k=1}^K \theta_k R_k, \nonumber\\
	&\text{s.~t. } \; \; \; \hat{\mathcal{C}}_2:  x_{n+1}^{\rm P} - x_{n}^{\rm P} \geq \Delta \ell,
\end{align}
\end{subequations}
for some convex coefficients $\theta_1, \ldots, \theta_K$. 
Since the system operates in TDMA mode, each device's data rate $ R_k$ can be optimized individually. Therefore, we solve $\hat{\mathcal{M}}_2$ separately for each device $k$ as follows:
\begin{align}
	 &\max_{\mathbf{x}_k^{\rm P}} R_k, \nonumber\\
	&\text{s.~t. } \; \; \; \hat{\mathcal{C}}_2:  x_{n+1, k}^{\rm P} - x_{n, k}^{\rm P} \geq \Delta \ell, \label{eq:M24}
\end{align}
where \( \mathbf{x}_k^{\rm P} = [x_{1,k}^{\rm P}, x_{2,k}^{\rm P}, \dots, x_{N,k}^{\rm P}] \) denotes the vector of scheduling locations for device \(k\).

\subsubsection{Gauss--Seidel coordinate update}
We start from the scalarized objective of $\hat{\mathcal{M}}_2$:
\[
\mathcal{J}(\mathbf{x}) \triangleq \sum_{k=1}^K \theta_k R_k(\mathbf{x})
= \sum_{k=1}^K \theta_k \log_2\!\Big(1 + a_k G_k(\mathbf{x})\Big),
\]
where $a_k \triangleq \frac{P_k \eta}{N\sigma^2}$ and $G_k(\mathbf{x})$ is the squared-magnitude channel gain in \eqref{gequation}.  
Since the scalar function $f(z) = \log_2(1+z)$ is concave for $z \ge 0$, Jensen's inequality yields
\begin{align}
\sum_{k=1}^K \theta_k \log_2\!\big(1 + a_k G_k\big)
\;\le\;
\log_2\!\Big( 1 + \sum_{k=1}^K \theta_k a_k G_k \Big).
\label{eq:Jensen}
\end{align}
where $\theta_k \ge 0$ and $\sum_{k=1}^K \theta_k = 1$.
Let $w_k \triangleq \theta_k a_k$ denote the combined user weights.
Because $\log_2(1+\cdot)$ is monotonically increasing, maximizing the Jensen upper bound (the right-hand side of \eqref{eq:Jensen}) is equivalent to
\begin{equation}\label{eq:weightedG}
\max_{\mathbf{x}} \; \sum_{k=1}^K w_k\, G_k(\mathbf{x})
\quad\text{s.t.}\quad
x_{n+1}^{\rm P} - x_{n}^{\rm P} \geq \Delta \ell, \;\; \forall n.
\end{equation}
Thus, we adopt the weighted-sum-of-gains in \eqref{eq:weightedG} as a tractable surrogate objective to optimize the placement.

\paragraph*{Coordinate isolation}
The objective in \eqref{eq:weightedG} couples all antenna positions through $G_k(\mathbf{x})$.  
We address this via a Gauss--Seidel-type approach \cite{dahlquist2008numerical}, where the antenna positions are updated sequentially.  
For a given index $n$, we treat $x_{n}$ as the only optimization variable and hold all other coordinates fixed. Define the \emph{single-antenna contribution function}
\begin{equation}\label{eq:Pi_def}
\Pi_k(x) \triangleq
\frac{\exp\!\big(-j(\kappa D_k(x) + \kappa_g \ell(x))\big)}{D_k(x)},
\end{equation}
where $D_k(x) \triangleq \sqrt{(x_k^{\mathrm{U}} - x)^2 + (y_k^{\mathrm{U}})^2 + d^2}$ is the free-space distance from user $k$ to coordinate $x$,  
$\ell(x)$ is the guided-path length from the feed-point to $x$,  
and $\kappa = 2\pi/\lambda$, $\kappa_g = 2\pi/\lambda_g$ are the free-space and guided wavenumbers, respectively.

Let $C_{k,n} \triangleq \sum_{\substack{m=1\\ m \neq n}}^N \Pi_k(x_m)$ be the (fixed) sum of contributions from all other antennas when antenna $n$ is excluded.
Then the user-$k$ channel gain can be written as
\begin{equation}
G_k(\mathbf{x}) = \big| \Pi_k(x) + C_{k,n} \big|^2.
\end{equation}
Plugging into \eqref{eq:weightedG} and collecting terms that depend on \(x\), the objective to maximize is
\begin{align}
\Phi_n(x)
&\triangleq \sum_{k=1}^K w_k \big|\Pi_k(x)+C_{k,n}\big|^2, \\
&= \sum_{k=1}^K w_k\Big( |C_{k,n}|^2 + 2\Re\{C_{k,n}^*\Pi_k(x)\} + |\Pi_k(x)|^2 \Big).\nonumber
\end{align}
Only the middle and last terms depend on \(x\). By dropping the constant term \(\sum_k w_k |C_{k,n}|^2\), we obtain 
\begin{equation}\label{eq:Phi_compact}
\Phi_n(x) \;=\; 2\,\Re\Big\{ \sum_{k=1}^K \zeta_{k,n}\, \Pi_k(x) \Big\}
\;+\; \sum_{k=1}^K \vartheta_{k,n}\, |\Pi_k(x)|^2,
\end{equation}
where $\zeta_{k,n} \triangleq w_k\, C_{k,n}^*$ and $\vartheta_{k,n} \triangleq w_k$. Hence, the coordinate update for antenna \(n\) becomes the 1-D optimization
\begin{align}\label{eq:scalar_opt}
x_n^\star \;=\;
\arg\max_{x\in\mathcal{X}_n}\; \Phi_n(x),
\end{align}
with feasible set
\[
\mathcal{X}_n
= \left\{ x \in [0,L_1] : |x - x_m| \ge \Delta\ell, \;\forall m \neq n \right\}.
\]
We solve \eqref{eq:scalar_opt} via an exhaustive grid search over $\mathcal{X}_n$.
In practice, we evaluate \eqref{eq:scalar_opt} over a uniform $L$-point grid on $[0,L_1]$,
discarding candidate locations that violate the minimum spacing constraint.
The feasible discrete set for coordinate $n$ is
\[
\mathcal{X}_n^{\mathrm{grid}}
\!=\! \left\{ \frac{\ell L_1}{L-1} \!\!\!\!\;:\;\!\!\! \ell = 0,\dots,L-1,\; |x - x_m| \ge \Delta\ell,\ \!\!\forall m \ne n \!\right\}\!.
\]
The coordinate update is then approximated by
\[
x_n^\star \;\approx\; \arg\max_{x \in \mathcal{X}_n^{\mathrm{grid}}} \; \Phi_n(x).
\]
\begin{algorithm}[t]
\caption{Two-tier Iterative Optimization}
\label{alg:final}
\begin{algorithmic}[1]
\State \textbf{Initialize:}  
    Set initial configuration $\mathbf{u}^{(0)}$, antenna positions $\mathbf{x}^{(0)}$,  
    tolerances $\epsilon_{\mathrm{outer}}, \epsilon_{\mathrm{inner}}$, and max iterations $T_{\mathrm{outer}}, T_{\mathrm{inner}}$.
\For{$t_{\mathrm{outer}} = 1, 2, \dots, T_{\mathrm{outer}}$}
    \State \textbf{(Outer-1)} Solve \(\mathcal{I}_1\) (Delay minimization) via convex programming.
    \State \textbf{(Outer-2)} Solve \(\mathcal{I}_2\) (scheduling) via linear programming.
    \State \textbf{(Outer-3)} Choose convex coefficients $\theta_1, \dots, \theta_K$.
    \State \textbf{(Outer-4)} Solve \(\tilde{\mathcal{I}}_3\) (power allocation).
    \Statex

    \State \textbf{Inner Loop: Gauss--Seidel Antenna Placement}
    \For{$t_{\mathrm{inner}} = 1, 2, \dots, T_{\mathrm{inner}}$}
        \For{$n = 1, 2, \dots, N$}
            \State Compute $C_{k,n} = \sum_{\substack{m=1\\m \neq n}}^N \Pi_k(x_m)$ for all $k$.
            \State Evaluate $\Phi_n(x)$ from \eqref{eq:Phi_compact}.
            \State Construct feasible grid set $\mathcal{X}_n^{\mathrm{grid}}$ excluding points violating spacing.
            \State Update $x_n \leftarrow \arg\max_{x \in \mathcal{X}_n^{\mathrm{grid}}} \Phi_n(x)$.
        \EndFor
        \If{antenna positions converge within $\epsilon_{\mathrm{inner}}$}
            \State \textbf{break}
        \EndIf
    \EndFor
\EndFor
\end{algorithmic}
\end{algorithm}
\subsection{Computational Complexity}
In general, the proposed two-tier optimization algorithm  alternates between an outer BCD procedure and an inner Gauss-Seidel antenna-placement refinement. Let $T_{\mathrm{out}}$ and $T_{\mathrm{in}}$ denote the numbers of outer and inner iterations, respectively. In the outer loop, the delay optimization, scheduling, and power allocation subproblems involve $\mathcal{O}(K)$ variables, where $K$ is the number of devices. Since these subproblems are convex (or linear) programs solved via interior-point methods, their worst-case complexity scales as $\mathcal{O}(K^3)$ per outer iteration.

Within each outer iteration, the inner loop sequentially updates the $N$ antenna positions. For each antenna, a one-dimensional grid search over $L$ candidate points is performed, and evaluating the objective at each grid point requires $\mathcal{O}(K)$ operations. Hence, one Gauss-Seidel sweep incurs $\mathcal{O}(N K L)$ complexity, and the total inner-loop cost is $\mathcal{O}(T_{\mathrm{in}} N K L)$. 

The overall complexity of the proposed algorithm is $\mathcal{O}\!\left(T_{\mathrm{out}}\left[K^3 + T_{\mathrm{in}} N K L\right]\right)$. When $L$, $N$, and $T_{\mathrm{in}}$ are moderate and do not scale with $K$, 
the dominant term is $\mathcal{O}(K^3)$, implying cubic scaling with respect to the number of devices.
\color{black}
\section{Numerical Result and Discussions}
In this section, we present numerical experiments to validate the derived expressions and evaluate the performance of \ac{pa} under an \ac{fl} setting. We begin by describing the simulation setup, followed by the presentation of the numerical results. 

\subsection{Communication Settings}
In the numerical simulations, we consider a rectangular area of size $30\,\text{m} \times 20\,\text{m}$, with the server located at the midpoint of the left boundary and a total of $K = 12$ users randomly placed within the region.
The height of the pinching antenna is set to $d = 3\,\text{m}$. 
The carrier frequency is set to $f_c = 28\,\text{GHz}$, the noise power spectral density is fixed at $-174\,\text{dBm/Hz}$, and the effective noise factor is set to $n_{\mathrm{eff}} = 1.44$.
Each device has a maximum transmit power budget of $P = 0.2\,\text{W}$ and a maximum local CPU frequency of $f_{\max} = 1.8\,\text{GHz}$. 
The computation load per data sample is set to $C_k = 10^{6}$ CPU cycles. 
The maximum computation energy budget per global round is constrained to $E_k^{\max} = 0.1\,\text{J}$ for all devices. The distance between the pinching elements is set to be larger than $\Delta\ell = \frac{\lambda}{2}$. We compare against three reference schemes:
\begin{itemize}
\item \textit{Conventional FL}: The conventional fixed-antenna baseline, where the base station is equipped with $N$ antennas at the center of the square region. The same scheduling and resource allocation procedure as in the proposed method is applied.
\item \textit{PASS-Uniform}: A \ac{pass} equipped with $N$ equally spaced antenna elements along the dielectric waveguide. The scheduling and power allocation follow the same optimization framework.
\item \textit{Perfect FL}: The ideal setting, where all devices participate in every round over perfect communication links. In this case, the server aggregates the local updates using standard FedAvg. Since the wireless channel is ideal, this benchmark reflects the best achievable learning performance.
\end{itemize}
\subsection{Learning Model and Datasets}
We evaluate \ac{pass}-aided FL on image-classification tasks with MNIST and CIFAR-10 datasets. 
The MNIST dataset contains 10 classes of grayscale handwritten digit images, with 60,000 training samples and 10,000 test samples \cite{xiao2017fashion}. The CIFAR-10 dataset consists of 10 classes of color images, with 6,000 images per class. Each class is split into 5,000 training images and 1,000 test images \cite{krizhevsky2009learning}. For MNIST, we adopt a four-layer convolutional network with two $3\times 3$ convolutional layers (32 and 64 channels, respectively), each followed by ReLU activation and $2\times 2$ max-pooling, and two fully connected layers with 128 hidden units. For CIFAR-10, we use a slightly deeper CNN with four $3\times 3$ convolutional layers (with 32, 32, 64 and 64 channels, respectively) grouped into two blocks, each block followed by $2\times 2$ max-pooling. A final $3\times 3$ convolution with 128 channels and an adaptive average pooling to a $4\times 4$ spatial resolution are applied to the output of the second block. The output is computed by a classifier with a 256-dimensional fully connected layer.  All hidden layers use ReLU activation.

To simulate the inherent non-uniformity of data in practical \ac{fl} systems, the training dataset is randomly shuffled and then partitioned across devices in a non-i.i.d. manner. We employ a Dirichlet-based label distribution with concentration parameter $\alpha = 0.35$. Specifically, for each class, its samples are allocated to the $K$ devices according to a Dirichlet$(\alpha)$ draw. Smaller values of $\alpha$ yield more skewed partitions, meaning that individual devices observe imbalanced and device-specific class distributions. This setup reflects realistic FL environments where clients hold statistically diverse local data \cite{marfoq2021federated}.

\subsection{Results and Discussions}
\par\textit{{1) PASS-Aided Training Accuracy}}: Fig.~\ref{fig:Fig11} illustrates the performance of FedPASS on the MNIST dataset and compares it with the benchmark schemes. The figure plots the test accuracy as a function of the number of communication rounds. As observed, the proposed FedPASS achieves a test accuracy that is remarkably close to the ideal Perfect FL benchmark, and it significantly outperforms both the conventional fixed-antenna baseline and the uniformly spaced PASS configuration. This improvement stems from the deployment of pinching antennas, which can be repositioned to locations closest to the active user in each time slot. By doing so, the system maximally reduces large-scale path loss and enhances the effective channel gain, thereby enabling more reliable and efficient model aggregation across rounds.

\par Fig.~\ref{fig:Fig1} shows the performance of the proposed PASS-aided framework on the CIFAR-10 dataset. Similar to the MNIST results, the PASS-optimized scheme consistently outperforms both benchmark methods. However, the performance gap between the proposed method and the benchmarks becomes more pronounced in this case. This arises because CIFAR-10 presents a substantially more challenging classification task with higher-dimensional inputs and more complex visual patterns compared to MNIST. Under such conditions, communication inefficiencies in the benchmark schemes lead to less reliable model updates, which amplifies their degradation in learning performance. 
In contrast, the PASS-aided design maintains strong uplink channels through optimized pinching-antenna placement, allowing more devices to successfully participate in each round. This reduces the aggregation error term and improves the convergence behavior, thereby preserving learning performance even in more challenging scenarios.
\begin{figure}[t!]
	\centering
    \includegraphics[height=0.33\textwidth]{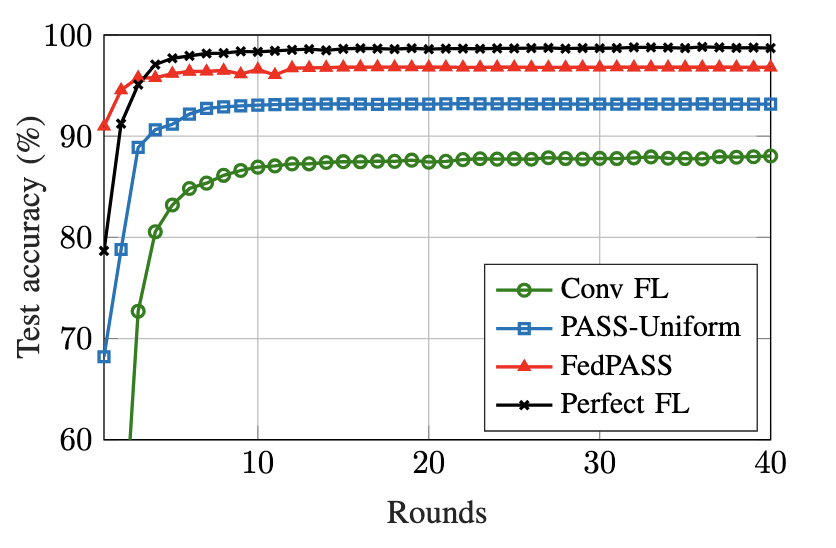}
	\caption{ Test accuracy vs the number of  rounds for MNIST dataset} 
\label{fig:Fig11}
\end{figure}
\begin{figure}[t!]
\centering
\includegraphics[height=0.33\textwidth]{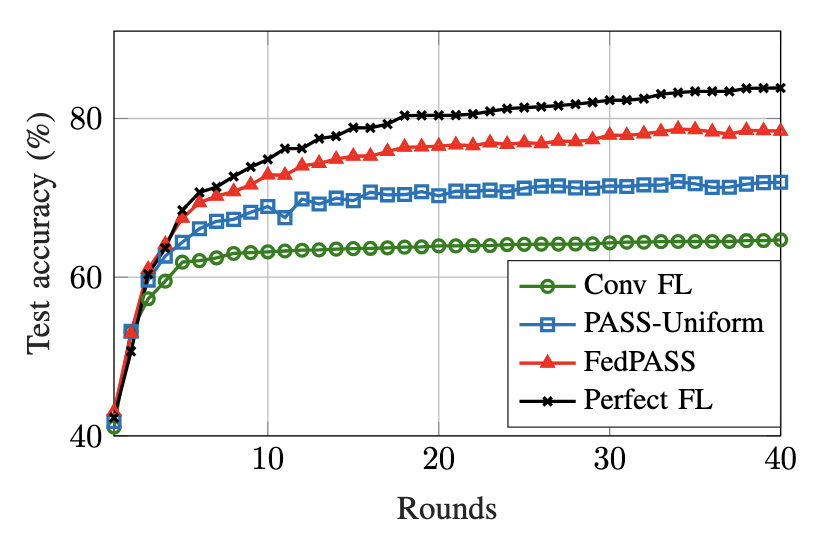}
\caption{Test accuracy vs the number of rounds for CIFAR dataset}
\label{fig:Fig1}
\end{figure}

\par \textit{{2) Total Latency}}: Fig.~\ref{fig:Fig2} illustrates the test accuracy versus total latency for the proposed FedPASS framework and the benchmark schemes, evaluated across different values of $K$. As the figure shows, FedPASS reaches the target accuracy approximately $6.4\times$ faster than the baseline methods. For all choices of $K$, the proposed PASS-FL design achieves the target accuracy significantly earlier than the conventional fixed-antenna baseline. Moreover, the performance gap becomes more pronounced as the number of users increases.

As $K$ grows, the conventional method experiences a significant increase in communication delay because more devices must complete their uplink transmissions before each round can proceed. This results in longer round duration and slower learning process. In contrast, the proposed PASS-FL method effectively mitigates large-scale path loss and improves the communication quality of users with weak channels through adaptive pinching-antenna placement. This leads to a substantially lower per-round latency and accelerates the completion of each global iteration. The comparison in Fig.~\ref{fig:Fig2} demonstrates that the proposed method scales more favorably with the number of participating users and consistently ensures faster round execution under identical system settings. These results highlight that the physical-layer enhancements introduced by PASS directly lead to a faster convergence rate in practice, achieving higher accuracy in significantly less time.
\begin{figure}[t!]
	\centering
    \includegraphics[height=0.30\textwidth]{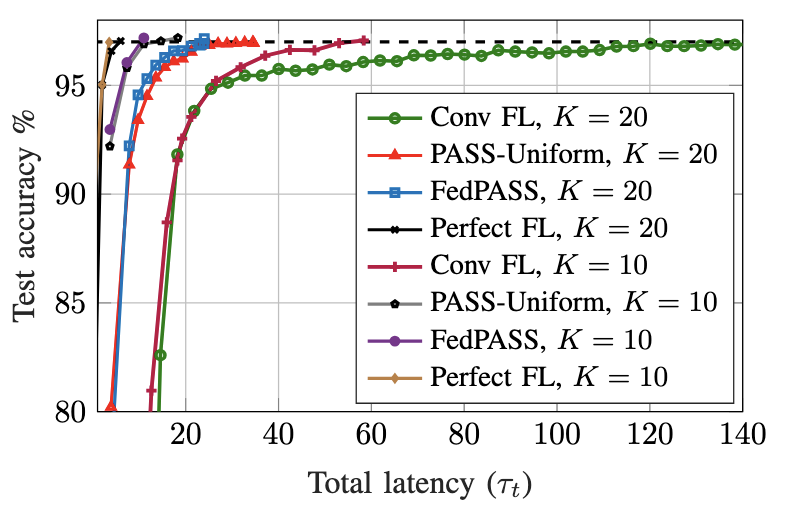}
	\caption{Accuracy vs total latency for different values of $K$.}
	\label{fig:Fig2}
\end{figure}
Fig.~\ref{fig:Fig4} illustrates the impact of the transmit-power budget on the test accuracy versus total latency for the considered FL schemes. The results clearly show that the proposed FedPASS approach maintains a substantial advantage across both low- and high-power regimes. When the transmit power is small (e.g., 
$P=0.05\, \text{W}$), the performance gap becomes particularly pronounced. In this regime, the conventional fixed-antenna baseline suffers from severe uplink bottlenecks due to limited channel quality, resulting in long transmission delays (imposed by the stragglers in the network) and degraded model aggregation. Similarly, the PASS-Uniform scheme cannot fully overcome these channel impairments, as its antenna placement is not tailored to individual device locations. In contrast, the FedPASS design adaptively positions the pinching antennas to strengthen the LoS component for each scheduled user. This targeted antenna placement significantly mitigates the effect of low transmit power by enhancing effective channel gains, thereby reducing per-round latency and enabling more reliable model uploads. As a result, PASS achieves noticeably faster convergence and higher accuracy. When the transmit power increases (e.g., $P=0.5\, \text{W}$), all schemes naturally benefit from improved SNR. However, the FedPASS design continues to perform superiorly, demonstrating that even when power is abundant, optimized antenna placement still accelerates convergence due to strong uplink channels.

\begin{figure}[t!]
	\centering
    \includegraphics[height=0.30\textwidth]{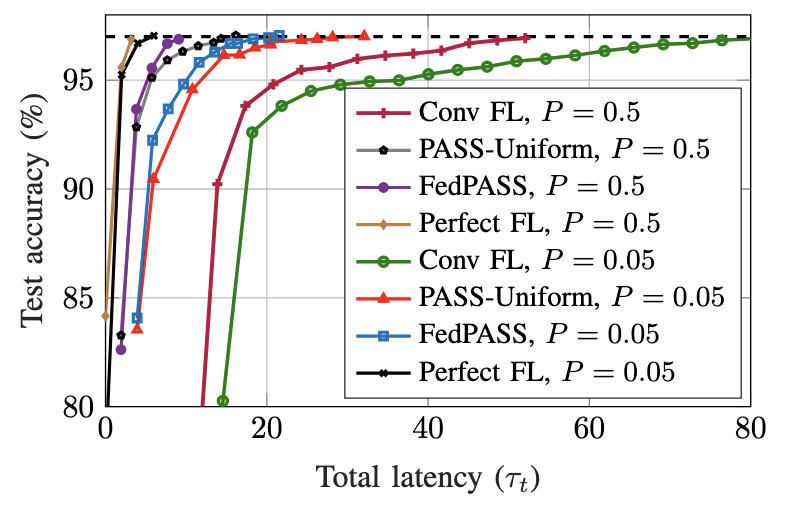}
	\caption{Accuracy vs total latency for various values of transmit power.}
	\label{fig:Fig4}
\end{figure}

\textit{{3) Latency-Learning Tradeoff}}:
Fig.~\ref{fig:Fig8} shows the Pareto-front that describes tradeoff between the total latency and the learning-performance metric for different transmit-power budgets.
The shape of the Pareto-front follows from the fact that both the objectives are set to be minimized: as the optimized total latency increases $F_{\rm learn}$ decreases, and vice versa. Note that smaller values of $F_{\rm learn}$ correspond to scheduling more devices per round, which improves learning quality but increases communication load and round duration. Conversely, relaxing the learning objective, i.e., larger $F_{\rm learn}$ allows the system to schedule fewer devices, reducing communication time and yielding lower latency. The comparison among different transmit-power budgets further reveals that higher $P$ shifts the entire Pareto front downward: with more available power, devices achieve higher uplink data rates, reducing communication time and enabling more devices to satisfy the energy constraint. As a result, higher-power configurations can operate at lower latency for the same learning level, or achieve better learning performance for the same latency budget. Lower transmit powers, by contrast, restrict feasible scheduling and yield significantly slower rounds, pushing the Pareto-front upward.

\begin{figure}[t!]
	\centering
    \includegraphics[height=0.30\textwidth]{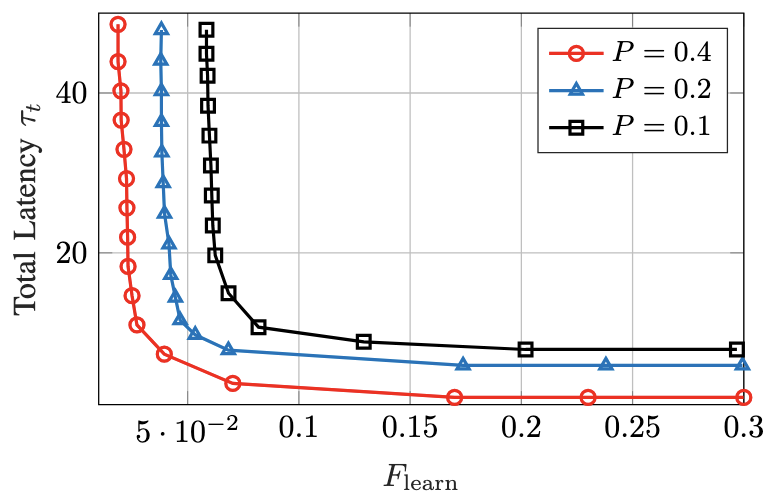}
	\caption{Trade-off between latency and learning quality under different power
		budget.}
	\label{fig:Fig8}
\end{figure}

\textit{{4) Impact of PASS Size}}: Fig.~\ref{fig:Fig5} depicts the total latency as a function of the number of pinching antennas $N$. As expected, the latency of the conventional fixed-antenna baseline remains constant since it does not depend on $N$. For both PASS-based schemes, however, the latency decreases as $N$ increases, following the fact that a larger number of antennas provides more spatial flexibility to place radiation points closer to the users. We further note that, the FedPASS scheme achieves the lowest latency across all values of $N$. Nevertheless, the latency reduction saturates beyond a certain threshold. This saturation occurs because the waveguide has a minimum required spacing between adjacent pinching antennas, imposed to avoid electromagnetic coupling and maintain independent radiation behavior. 
\begin{figure}[t!]
	\centering
    \includegraphics[height=0.30\textwidth]{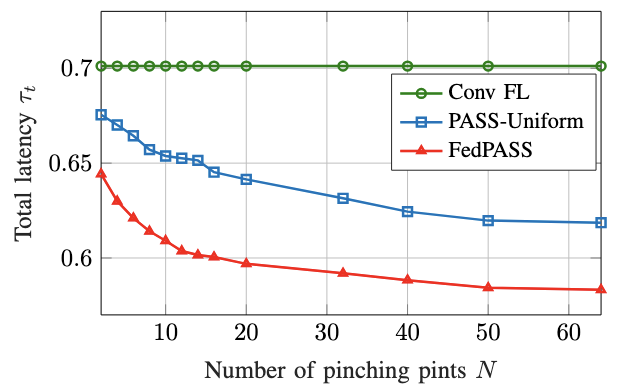}
	\caption{Total latency vs the number of pinching antenna}
	\label{fig:Fig5}
\end{figure}
 \begin{figure}[t!]
 	\centering
    \includegraphics[height=0.30\textwidth]{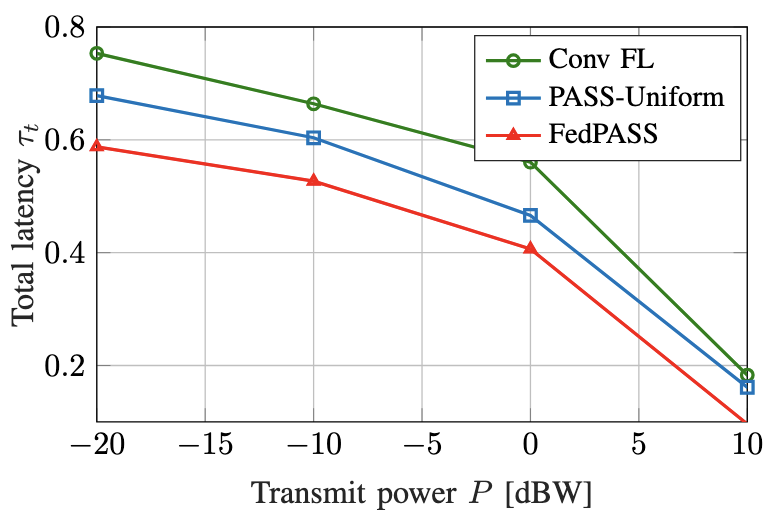}
 	\caption{Total latency vs the transmit power}
 	\label{fig:Fig6}
 \end{figure}
\par \textit{{5) Impact of Transmit Power}}: Fig.~\ref{fig:Fig6} illustrates the impact of the transmit-power budget $P$ on the total latency. As shown, increasing the transmit power leads to a reduction in latency for all wireless schemes. This occurs because higher power improves the uplink SNR, resulting in higher achievable data rates and shorter model-upload times per device.
Among the approaches, the conventional fixed-antenna baseline exhibits the highest latency, since it cannot enhance the LoS component or compensate for users with inherently weak channels. The PASS-Uniform scheme moderately improves latency by enabling radiation points along the waveguide, but its non-optimized antenna placement limits the achievable gains. In contrast, FedPASS achieves the lowest latency across all power levels. By adaptively positioning the pinching antennas, FedPASS strengthens uplink channels even for users located in challenging positions, enabling much faster transmissions. 
\begin{figure}[t!]
	\centering
    \includegraphics[height=0.30\textwidth]{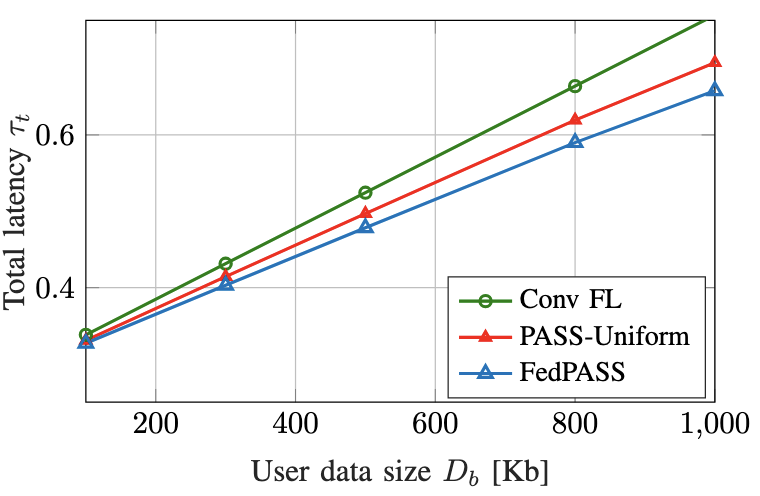}
	\caption{ Total latency vs the user data size}
	\label{fig:Fig7}
\end{figure}
\par \textit{{6) Impact of Local Data-Size}}: Fig.~\ref{fig:Fig7} illustrates how the total latency varies with the data size $D_b$ that each device must upload during the uplink phase. As expected, increasing $D_b$ leads to a growth in total latency across all schemes. Since each device must transmit a larger number of bits, its individual uplink time increases accordingly. Noting that uploading larger updates also consumes more energy, fewer devices are able to meet the energy constraint in each round. As a result, fewer users are scheduled and the system needs more time to complete each global iteration. These two effects together cause the total latency to grow as $D_b$ becomes larger. Similar to earlier experiments, FedPASS consistently achieves the lowest latency among the approaches due to its ability to strengthen uplink channels via optimized pinching-antenna placement, partially mitigating the communication burden imposed by larger model sizes.

\section{Conclusion}
This paper introduces FedPASS, a framework that leverages PASS to accelerate the training process of wireless FL. We formulated a bi-objective mixed-integer optimization problem that jointly minimizes training latency and an upper bound on the FL optimality gap. A two-tier iterative algorithm was developed to solve the resulting problem through coordinated updates of scheduling, power control, and PA placement. Numerical results on the MNIST and CIFAR-10 datasets demonstrate that FedPASS reduces the total training latency by up to $6.4\times$ compared to benchmark schemes, while achieving learning accuracy close to that of idealized baselines. 
Future work may explore asynchronous FL, more advanced PA activation strategies, and large-scale distributed optimization.

\appendices
\section{}
Given (A.2), by $L$-smoothness we have
\[
F(\boldsymbol{\omega}_t) \leq F(\boldsymbol{\omega}_{t-1}) + \langle \nabla F(\boldsymbol{\omega}_{t-1}), \boldsymbol{\omega}_t - \boldsymbol{\omega}_{t-1} \rangle + \frac{L}{2} \|\boldsymbol{\omega}_t - \boldsymbol{\omega}_{t-1}\|_2^2.
\]
After some calculation, we can write the global update as $\boldsymbol{\omega}_t
=\boldsymbol{\omega}_{t-1}-\frac{\vartheta}{L}\,\mathbf{g}_{t-1}$, where
\begin{align*}
    \mathbf{g}_{t-1}
=\frac{1}{D(\mathbf{s}^t)}\sum_{k=1}^K s_k^t |\mathcal{D}_k|
\left(\frac{1}{\vartheta}\sum_{j=0}^{\vartheta-1}\nabla F_k(\boldsymbol{\omega}_k^{(j)})\right).
\end{align*}
Hence, we obtain
\begin{align}\label{eq:feq}
    F(\boldsymbol{\omega}_t)
\le
F(\boldsymbol{\omega}_{t-1})
-\frac{\vartheta}{L}\big\langle \nabla F(\boldsymbol{\omega}_{t-1}),\,\mathbf{g}_{t-1}\big\rangle
+\frac{\vartheta^2}{2L}\big\|\mathbf{g}_{t-1}\big\|_2^2.
\end{align}
Let \(\mathbf{g}_{t-1} = \nabla F(\boldsymbol{\omega}_{t-1}) - \mathbf{e}_{t-1} + \mathbf{d}_{t-1}\), where $ \mathbf{e}_{t-1}$ and $\mathbf{d}_{t-1}$ are gradient-sampling error and the local-drift term, respectively. $ \mathbf{e}_{t-1}$ is defined in \eqref{eq:e} and $\mathbf{d}_{t-1}$ can be written as
\begin{align}\label{eq:d}
    \mathbf{d}_{t-1}
=\frac{1}{D(\mathbf{s}^t)}\sum_{k=1}^K s_k^t |\mathcal{D}_k|
\left(\frac{1}{\vartheta}\sum_{j=0}^{\vartheta-1}
\big[\nabla F_k(\boldsymbol{\omega}_k^{(j)})-\nabla F_k(\boldsymbol{\omega}_{t-1})\big]\right).
\end{align}
Applying inequality
$-\langle a,b\rangle \ge -\tfrac{1}{2\gamma}\|a\|^2 - \tfrac{\gamma}{2}\|b\|^2$ \cite{lang2012algebra}
with $\gamma=2$ to both cross terms, we can bound inner product as 
\begin{align*}
    \big\langle \nabla F(\boldsymbol{\omega}_{t-1}),\,\mathbf{g}_{t-1}\big\rangle
\ge
\frac{1}{2}\big\|\nabla F(\boldsymbol{\omega}_{t-1})\big\|_2^2
-\frac{1}{2}\big\|\mathbf{e}_{t-1}\big\|_2^2
-\frac{1}{2}\big\|\mathbf{d}_{t-1}\big\|_2^2.
\end{align*}
Substituting the above equation into \eqref{eq:feq} yields
\begin{align}\label{eq:F22}
    F(\boldsymbol{\omega}_t)
\le
F(\boldsymbol{\omega}_{t-1})
&-\frac{\vartheta}{2L}\big\|\nabla F(\boldsymbol{\omega}_{t-1})\big\|_2^2
+\frac{\vartheta}{2L}\big\|\mathbf{e}_{t-1}\big\|_2^2 \nonumber \\
&+\frac{\vartheta}{2L}\big\|\mathbf{d}_{t-1}\big\|_2^2.
\end{align}
\par By Assumption~A.4 and $L$-smoothness, we have
\begin{equation}
\big\|\nabla F_k(\boldsymbol{\omega}_k^{(j)})-\nabla F_k(\boldsymbol{\omega}_{t-1})\big\|_2
\le
L\,\big\|\boldsymbol{\omega}_k^{(j)}-\boldsymbol{\omega}_{t-1}\big\|_2.
\end{equation}
Using \eqref{eq:d}, and Jensen/Cauchy–Schwarz on the weighted sum over scheduled devices, and summing over the $\vartheta$ local steps yields an extra factor of $\vartheta^2$. Hence we obtain the squared-norm bound
\begin{align}\label{eq:d3}
\big\|\mathbf{d}_{t-1}\big\|_2^2
\le
\frac{\vartheta^2\,\varepsilon^2}{|\mathcal{D}|\,D(\mathbf{s}^t)}.
\end{align}
Finally, substituting \eqref{eq:d3} into \eqref{eq:F22} completes the proof.
\section{}

Applying Lemma~\ref{lem:1} and invoking the Polyak–Łojasiewicz inequality in A.3,
the one-step descent yields
\begin{align*}
O_t \le \left(1-\frac{\delta\vartheta}{L}\right) O_{t-1}
+\frac{\vartheta}{2L}\|\mathbf{e}_{t-1}\|_2^2
+\frac{\vartheta^3\varepsilon^2}{2L|\mathcal{D}|D(\mathbf{s}^t)}.
\end{align*}
The last term represents the contribution of the local-drift error.
Since $D(\mathbf{s}^t)\le|\mathcal{D}|$, we upper-bound it by 
$\frac{\vartheta^3\varepsilon^2}{2L|\mathcal{D}|^2}$, which is included in $A_t$ as a drift component. Using the boundedness of sample gradients and unrolling the recursion over $T$ rounds gives the stated upper bound.

\bibliographystyle{IEEEtran}
\bibliography{icme2022template}
\end{document}